\documentclass[A4,11pt]{amsart}
\usepackage{amsthm,amsmath,amssymb,eucal}
\usepackage{geometry}
\usepackage{resizegather}
\usepackage{graphicx}
\usepackage{color}
\usepackage[normalem]{ulem}
\usepackage[latin1]{inputenc}

\newtheorem{claim}{Claim}[section]
\newtheorem{theorem}[claim]{Theorem}

\newtheorem{lemma}[claim]{Lemma}
\newtheorem{remark}[claim]{Remark}

\newtheorem{definition}[claim]{Definition}
\newtheorem{corollary}[claim]{Corollary}

\newtheorem{thmx}{Theorem}

\newcommand{\soutg}{\bgroup\markoverwith{\textcolor{green}{\rule[.5ex]{2pt}{1pt}}}\ULon}
\newcommand{\soutb}{\bgroup\markoverwith{\textcolor{blue}{\rule[.5ex]{2pt}{1pt}}}\ULon}
\newcommand{\soutr}{\bgroup\markoverwith{\textcolor{red}{\rule[.5ex]{2pt}{1pt}}}\ULon}

\newcommand{\bo}{{\rm O}}

\newcommand{\ds}{\displaystyle}
\newcommand{\dsum}{\ds\sum}
\newcommand{\dint}{\ds\int}

\DeclareMathOperator*{\detz}{det}
\newcommand{\eqskip}{ \vspace*{2mm}\\ }
\newcommand{\fr}[2]{\frac{\ds #1}{\ds #2}}
\newcommand{\dprod}{\ds\prod}
\DeclareMathOperator*{\re}{Re}
\newcommand{\tto}{ \mathcal{T} }
\newcommand{\ph}{P_{n}}

\title[The determinant of polyharmonic operators]{The determinant of one-dimensional polyharmonic operators of arbitrary order}

\author{Pedro Freitas} 
\author{Ji\v{r}\'{\i} Lipovsk\'{y}}

\address{Departamento de Matem\'atica, Instituto Superior T\'ecnico, Universidade de Lisboa, Av. Rovisco Pais 1,
P-1049-001 Lisboa, Portugal {\rm and}
Grupo de F\'isica Matem\'{a}tica, Faculdade de Ci\^encias, Universidade de Lisboa,
Campo Grande, Edif\'icio C6, P-1749-016 Lisboa, Portugal}
\email{psfreitas@fc.ul.pt}
\address{Department of Physics, Faculty of Science, University of Hradec Kr\'alov\'e, Rokitansk\'eho 62,
500\,03 Hradec Kr\'alov\'e, Czechia}
\email{jiri.lipovsky@uhk.cz}

\begin{document}

\begin{abstract} We obtain an explicit expression for the regularised spectral determinant of the polyharmonic operator
$\ph=(-1)^{n} (\partial_x)^{2n}$ on $(0,T)$ with Dirichlet boundary conditions and $n$ a positive integer, and
show that it satisfies the asymptotics $\log{(\det \ph)} = -n^2 \log{n} + \left[\fr{7\zeta(3)}{2\pi^2}+
\fr{3}{2}+\log\left(\fr{T}{4}\right)\right] n^2 + \bo(n)$ for large $n$. This is a
consequence of sharp upper and lower bounds for $\log{(\det \ph)}$ valid for all $n$ and which coincide in the terms up to order $n$.
These results form the basis to analyse more general operators with nonconstant coefficients and show that the corresponding determinants
have a similar asymptotic behaviour.
\end{abstract}

\maketitle

\section{Introduction}

Let $\tto$ be an elliptic differential operator of order $m$ with discrete spectrum denoted by
$
 \lambda_{1}\leq \lambda_{2}\leq\cdots.
$
The regularised spectral determinant of $\tto$ has its roots in the 1971 work of Ray and Singer and is one possible way of
making sense of the infinite product of the numbers $\lambda_{k}$, $k=1,2,\cdots$~\cite{rasi} -- see also~\cite{geya,MP}. In
this approach, one begins by defining the spectral zeta function associated with the sequence of eigenvalues of the operator
$\tto$ by
\[
 \zeta_{\tto}(s) = \dsum_{k=1}^{\infty} \fr{1}{\lambda_{k}^{s}}
\]
on some half-plane $\re(s)>s_{0}$. If $\zeta_{\tto}$ has a meromorphic extension to the whole of the complex plane which
is analytic at zero, it is then possible, by analogy with an identity that holds true in case of a finite number of eigenvalues,
to define the determinant of the operator $\tto$ by
\[
\detz(\tto) := \exp\left(-\zeta'(0)\right).
\]
Applying the above definition to the operator $-u_{D}''$ on the interval $[0,\pi]$ with Dirichlet boundary conditions, for instance,
produces $\detz(-u_{D}'')=2\pi$. However, and while allowing us to make sense of the determinant in an infinite-dimensional
setting, the above expression will not be simple to evaluate explicitly in general since,
if nothing else, eigenvalues of operators are not known explicitly. Examples where explicit expressions have been found
for the determinant are the Dirichlet Laplacian on balls~\cite{bgke}, triangles and some other polygons~\cite{aursal},
and Sturm-Liouville operators~\cite{Levit-Smilansky,Gesztesy-Kirsten}.

Some of the above results are surprising in the sense that, as is the case for triangles, there is an explicit formula for
the determinant although this is not the case for individual eigenvalues in general. On the other hand, even when they are
known explicitly the complexity of the expression for the determinant may become quite unmanageable, as may be seen from the
expressions given up to dimension six in~\cite{bgke}, or those computed in~\cite{frei} for the isotropic harmonic oscillator.
In this last example, and although a recurrence formula was given which permitted the calculation of the determinant for any
dimension, starting from the (simple) expressions for dimensions one and two, but the formulas obtained soon become quite
complicated. Because of this growing complexity, the approach used in~\cite{frei} was to study the asymptotic behaviour of
the determinant as the dimension became large. As shown it that paper, the first two terms in this asymptotic expansion
already provide quite an accurate approximation even for low dimension values, yielding an alternative way to tackle this
type of problem.

The purpose of this paper is to illustrate how a similar approach may be applied to a case where now, instead of the dimension,
we vary the order of the operator. More precisely, we first consider the polyharmonic operator $\ph := (-1)^{n} (\partial_x)^{2n}$ on a
bounded interval $(0,T)$, for some positive number $T$, together with Dirichlet boundary conditions, that is, $u$ and its
derivatives up to order $(n-1)$ vanish at both endpoints of the interval. In this instance, eigenvalues are roots of transcendental
equations which get more complex as $n$ increases, and which cannot be determined explicitly.
In spite of this, it is still possible to write down an explicit expression for the corresponding determinant, 
using the general result for determinants of matrix operators studied in~\cite{BFK}. This is, however, still quite
complicated -- see Theorem~\ref{thm:h}. Our approach is to use this as a starting point, and first determine a
much simplified closed form for the determinant as a function of $n$. This type of problem is now closer
to that of evaluating the determinants of families of matrices with a special form, such as (truncated) Toeplitz matrices,
and study the behaviour of the corresponding determinants as the dimension of the matrix grows to infinity. In fact, one of the steps
in obtaining a general form for the determinant of the operator $\ph$ is to show that it may be written as the product of two
determinants, one of which may be reduced to the calculation of the determinant of a Vandermonde matrix.
We then obtain upper and lower bounds for this resulting expression, which are precise enough to allow us to derive its asymptotic
behaviour as the order becomes large. Our first main result is then the following

\begin{thmx}\label{thm:h2}
The determinant of the polyharmonic operator $\ph=(-1)^{n} (\partial_x)^{2n}$ defined on the interval $(0,T)$ with Dirichlet
boundary conditions is given by
\[
  \det \ph = \left(\fr{T}{2}\right)^{n^2} \fr{(4n)^n}{\prod_{k=1}^{n-1}\left[\sin^{2(n-k)}{\left(\frac{k\pi}{2n}\right)}\right]}
  \prod_{k=0}^{n-1} \fr{k!}{(n+k)!}\,.
\]

Furthermore, for large values of $n$ it satisfies
\begin{equation}\label{logdet}
  \log{(\det \ph)} = -n^2 \log{n} + \left[\frac{7\zeta(3)}{2\pi^2}+\frac{3}{2}+\log{\frac{T}{4}}\right] n^2 + \bo(n)\,.
\end{equation}
\end{thmx}

As stated above, the asymptotic behaviour desribed in the theorem is a consequence of sharp bounds which are given
in Section~\ref{sec:ulbounds}. We remark that this is quite different from what happens with the determinant of the same polyharmonic
operator when considered with Navier boundary conditions, that is, when the function and its derivatives of even order up to
$2n-2$ vanish at both endpoints of the interval. The resulting operator, say $N_n$, is now the $n^{\rm th}$ power of the
Dirichlet Laplacian, and so its eigenvalues may be computed explicitly and are given by
\[
 \gamma_{k} = \left( \fr{k \pi}{T}\right)^{2n}.
\]
The associated zeta function becomes
\[
 \zeta_{N_n}(s) = \dsum_{k=1}^{\infty} \left( \fr{T}{k\pi}\right)^{2ns},
\]
yielding $\zeta_{N_n}'(0) = -n \log(2T)$ and
\[
 \det(N_{n}) = 2^n T^n.
\]
We thus see that in this case the behaviour of the determinant with the order of the operator actually depends in a critical
way on the length of the interval, with the separation in behaviour taking place at length $T=1/2$. Although this is not comparable to the
case of Dirichlet boundary conditions, due to the dominant term there being of the form $\mathrm{e}^{-n^2\log(n)}$, we point out that the
coefficient of $n^{2}$ in~\eqref{logdet} becomes negative for $T<4\,\mathrm{e}^{-(7\zeta(3)/(2\pi^2)+3/2)}\approx 0.582758$.

It is not unrealistic to expect that changing the operator $P_{n}$ by adding lower order terms should not change the asymptotic behaviour of
the determinant as its order goes to infinity. Indeed, provided that the perturbing terms are of order $m$ with $m=m(n) \leq n$, it is
possible to show that the difference between the original determinant and that of this more general operator is of lower order in $n$. More
precisely, our second main result is the following:

\begin{thmx}\label{thm:perturb}
 Let
 \[H_{n} = (-1)^{n} (\partial_x)^{2n}+\sum_{j=0}^m q_j(x)(\partial_x)^{j}\]
 be the polyharmonic operator defined on the interval $(0,T)$ together with Dirichlet boundary conditions,
and where $q_{j}\in C^\infty([0,T])$ are complex functions defined on the interval $[0,T]$, and $m$ is either fixed or dependent on $n$ such that $m(n)\leq n$.
Then, the determinant of the operator $H_{n}$ satisfies
\begin{equation*}
  \log{(\det H_n)} = \log{(\det \ph)}+ {\rm O}\left(1/n\right)\,.
\end{equation*}
as $n$ goes to $\infty$. In particular, it has the same asymptotic behaviour as that of $\log(\det P_{n})$ given in~\eqref{logdet}.
\end{thmx}

The conditions in the above result may be relaxed somewhat. For instance, the coefficient functions $q_{j}$ may be allowed to depend on $n$, provided
they remain uniformly bounded.
Although our methods do not allow us to prove Theorem~\ref{thm:perturb} for $m(n)>n$, we believe that a similar result will hold also for higher values of $m$,
provided that $q_{2n}(x) = q_{2n-1}(x) = 0$. While the restriction in the $q_{2n}$ term should be expected, the reason for the restriction
on $q_{2n-1}$ is more subtle and seems to be mainly due to a question of the phase of the determinant. Concerning this, it is interesting
to note that having complex coefficients for at least the terms up to order $n$, and possibly a set of corresponding non-real eigenvalues, does
not affect the first terms in the asymptotics.

\section{A first expression for the determinant}
For the reader's convenience we recall a theorem from \cite{BFK} which gives the determinant for a more general type of operator. 

\begin{definition}\label{def1}
Let us assume an operator acting as $\tto = \dsum_{k=0}^{2n} a_k(x)(-i)^k \frac{\mathrm{d}^k}{\mathrm{d}x^k}$, where $a_k(x)$ are
complex-valued $r\times r$ matrices, which depend smoothly on $x\in[0,T]$. Let us assume that $a_{2n}(x)$ is not singular and
that there exists a principle angle $\theta$ such that $\mathrm{spec}(a_{2n}(x))\cup \{\rho \mathrm{e}^{i\theta}\in\mathbb{C}:0
\leq \rho <\infty\} = \emptyset$ ($\mathrm{spec}$ denotes the spectrum of the matrix). Let us assume boundary conditions at $x=0$ and $x=T$
$$
  \dsum_{k=0}^{\alpha_j} b_{jk} u^{(k)}(T) = 0\,,\quad \dsum_{k=0}^{\beta_j} c_{jk} u^{(k)}(0) = 0\,,\quad \mathrm{for\,}1\leq j\leq n\,.
$$
Here $b_{jk}$ and $c_{jk}$ are $r\times r$ matrices with $b_{j,\alpha_j} = c_{j,\beta_j} = I$ ($I$ denotes the identity matrix) and
$\alpha_j$ and $\beta_j$ satisfy
\begin{eqnarray*}
  0\leq \alpha_1 <\alpha_2 <\dots < \alpha_n\leq 2n-1\,,\\
  0\leq \beta_1 <\beta_2 <\dots < \beta_n\leq 2n-1\,.
\end{eqnarray*}
For the Dirichlet boundary conditions we have
$\alpha_{\mathrm{D}j} := \beta_{\mathrm{D}j} = j-1$ and 
$$
  b_{\mathrm{D},jk} = c_{\mathrm{D},jk} := \left\{\begin{matrix}I & \mathrm{for\ }1\leq j \leq n\,,\ k=j-1\,, \\ 0 & \mathrm{otherwise}\end{matrix}\right.\,.
$$

Furthermore, we define $|\alpha| = \dsum_{j=1}^n \alpha_j$ and $|\beta| = \dsum_{j=1}^n \beta_j$ and the $2n\times 2n$ matrices $B$ and $C$, whose entries with indices $1\leq j\leq 2n$, $0\leq k \leq 2n-1$ 
\begin{eqnarray*}
   B_{jk}&:=& \left\{\begin{matrix}b_{jk} & \mathrm{for\ }1\leq j \leq n \mathrm{\ and \ }0\leq k \leq \alpha_j \\ 0 & \mathrm{otherwise}\end{matrix}\right.\,,\\
   C_{jk}&:=& \left\{\begin{matrix}c_{j-n,k} & \mathrm{for\ }n+1\leq j \leq 2n \mathrm{\ and \ }0\leq k \leq \beta_{j-n} \\ 0 & \mathrm{otherwise}\end{matrix}\right.\,.
\end{eqnarray*}
are $r\times r$ matrices. We define
$$
  g_\alpha := \frac{1}{2}\left(\frac{|\alpha|}{n}-n+\frac{1}{2}\right)\,,\quad h_\alpha = \mathrm{det\,}\begin{pmatrix}w_1^{\alpha_1}& \dots & w_n^{\alpha_1}\\ \vdots & \ddots & \vdots \\ w_1^{\alpha_n}& \dots & w_n^{\alpha_n}\end{pmatrix}
$$
with $w_k = \mathrm{exp}(\frac{2k-n-1}{2n}\pi i)$. Similarly, one defines $g_\beta$ and $h_\beta$. If $\lambda_j$, $1\leq j\leq r$ are eigenvalues of a matrix $a$, we denote by 
$$
  (\mathrm{det\,}a)_\theta^{g_\alpha} := \prod_{j=1}^r |\lambda_j|^{g_\alpha} \mathrm{e}^{i g_\alpha \mathrm{arg\,}\lambda_j}\,,
$$
where $\theta-2\pi < \mathrm{arg\,}\lambda_j < \theta$ and $\theta$ is the principal angle of the matrix $a$. 

Finally, we define a $2n\times 2n$ matrix $Y(x) = (y_{k\ell}(x))$, with its entries being $r\times r$ matrices $y_{k\ell}(x):=y_\ell^{(k)}(x)$,
$0\leq k,\ell\leq 2n-1$. Here $y_\ell(x)$ are solutions to the Cauchy problem $\tto y_\ell(x) = 0$ with $y_\ell^{(k)}(0) = \delta_{k\ell}I$. 
\end{definition}

\begin{theorem}\label{thm:bfk}\emph{(Burghelea, Friedlander, Kappeler \cite{BFK})} 
The determinant for the operator $\tto$ is equal to
$$
  \det_\theta \tto = K_\theta \mathrm{exp}\left(\frac{i}{2}\int_0^T \mathrm{Tr\,}(a_{2n}^{-1}(x)a_{2n-1}(x))\,\mathrm{d}x\right)\,\mathrm{det\,}(BY(T)-C)\,,
$$
where
$$
  K_\theta = [(-1)^{|\beta|}(2n)^n h_\alpha^{-1} h_\beta^{-1}]^r (\mathrm{det\,}a_{2n}(0))_\theta^{g_\beta} (\mathrm{det\,}a_{2n}(T))_\theta^{g_\alpha}\,.
$$
\end{theorem}

\begin{theorem}\label{thm:h}
Let us consider the operator acting as
$\ph u = (-1)^n (\partial_x)^{2n} u$ on $(0,T)$ with the boundary conditions $u^{(s)}(0)=u^{(s)}(T) = 0$ for $s = 0, 1,\dots, n-1$.

Then the determinant for this operator is 
$$
  \det \ph = (2n)^n |h_\alpha|^{-2}\mathrm{det\,}
\begin{pmatrix}
\frac{1}{n!}T^n & \frac{1}{(n+1)!}T^{n+1} & \cdots & \frac{1}{(2n-1)!}T^{2n-1}\\
\frac{1}{(n-1)!}T^{n-1} & \frac{1}{n!}T^{n} & \cdots & \frac{1}{(2n-2)!}T^{2n-2}\\
\vdots & \vdots & \ddots & \vdots\\
\frac{1}{2!}T^2 & \frac{1}{3!}T^3 & \cdots & \frac{1}{(n+1)!}T^{n+1} \\
T & \frac{1}{2!}T^2 & \cdots & \frac{1}{n!}T^n 
\end{pmatrix}\,,
$$
where 
$$
  h_\alpha = \mathrm{det\,}
\begin{pmatrix}
  1 & 1 & \cdots & 1\\
  \mathrm{e}^{\frac{1-n}{2n}\pi i} & \mathrm{e}^{\frac{3-n}{2n}\pi i} & \cdots & \mathrm{e}^{\frac{2n-1-n}{2n}\pi i}\\
  \mathrm{e}^{\frac{1-n}{2n}2\pi i} & \mathrm{e}^{\frac{3-n}{2n}2\pi i} & \cdots & \mathrm{e}^{\frac{2n-1-n}{2n}2\pi i}\\
  \mathrm{e}^{\frac{1-n}{2n}3\pi i} & \mathrm{e}^{\frac{3-n}{2n}3\pi i} & \cdots & \mathrm{e}^{\frac{2n-1-n}{2n}3\pi i}\\
  \vdots & \vdots & \ddots & \vdots\\
  \mathrm{e}^{\frac{1-n}{2n}(n-1)\pi i} & \mathrm{e}^{\frac{3-n}{2n}(n-1)\pi i} & \cdots & \mathrm{e}^{\frac{2n-1-n}{2n}(n-1)\pi i}
\end{pmatrix}\,.
$$
\end{theorem}
\begin{proof}
The theorem is a direct consequence of Theorem~\ref{thm:bfk}. In our case, we have $r = 1$, $a_{2n} = 1$, $a_{k} = 0$ for $0\leq k \leq 2n-1$, $\alpha_j = \beta_j = j-1$, $|\alpha| = |\beta| = \frac{n(n-1)}{2}$. $B$ and $C$ are $2n\times 2n$ matrices.
$$
  B = \begin{pmatrix}I_n & 0\\ 0 & 0 \end{pmatrix}\,,\quad C = \begin{pmatrix} 0 & 0\\ I_n & 0 \end{pmatrix}\,,
$$ 
where $I_n$ is $n\times n$ identity matrix and $0$ is $n\times n$ matrix with all entries equal to zero. Furthermore, $g_\alpha = g_\beta = -\frac{n}{4}$, $h_\alpha = h_\beta$ is given by the expression in the statement of the theorem, $\lambda_1 = 1$, $(\mathrm{det}\,a_{2n}(x))_\theta^{g_\alpha}=(\mathrm{det}\,a_{2n}(x))_\theta^{g_\beta}=\mathrm{e}^{-\frac{in}{4}0} = 1$, $K_\theta = (-1)^{\frac{n(n-1)}{2}}(2n)^n h_\alpha^{-2}$, $\mathrm{Tr\,}(a_{2n}^{-1}(x) a_{2n-1}(x)) = 0$, $y_\ell(x) = \frac{1}{\ell!}x^{\ell}$, $y_{k\ell} = \frac{1}{(\ell-k)!}x^{(\ell-k)}$. 
$$
  Y(x)=\begin{pmatrix}1 & x & \frac{x^2}{2!} & \dots & \frac{x^{2n-1}}{(2n-1)!}\\0 & 1 & x &  \dots & \frac{x^{2n-2}}{(2n-2)!}\\ \vdots & \vdots & \vdots & \ddots & \vdots\\ 0 & 0 & 0 & \dots & 1\end{pmatrix}
$$
\begin{equation}
\begin{array}{lll}
  \mathrm{det\,}(BY(T)-C) & = & \mathrm{det\,}\begin{pmatrix}1 & T & \dots & \frac{T^{n-1}}{(n-1)!} & \frac{T^{n}}{n!} & \dots &\frac{T^{2n-1}}{(2n-1)!}\\
                                          \vdots & \vdots & \ddots & \vdots & \vdots & \ddots & \vdots\\
										  0 & 0 & \dots & 1& T & \dots & \frac{T^n}{n!}\\
										  \multicolumn{4}{c}{-I_n} & \multicolumn{3}{c}{0}\end{pmatrix}\eqskip
  & = & \mathrm{det\,}\begin{pmatrix} \frac{T^{n}}{n!} & \dots &\frac{T^{2n-1}}{(2n-1)!}\\ \vdots & \ddots & \vdots \\ T & \dots & \frac{T^n}{n!}\end{pmatrix}\,.\label{eq:byc}
  \end{array}
\end{equation}
The result follows from Theorem~\ref{thm:bfk} with the use of $\bar{h}_\alpha = (-1)^{\frac{n(n-1)}{2}}h_\alpha$.
\end{proof}

Note that the determinant of the operator $\ph$ does not depend on the principal angle $\theta$ as long as $\theta \in (0,2\pi)$, therefore,
in this paper we omit the index $\theta$ from the expression for the determinant.

It can also be found multiplying the rows of the matrix under the determinant by certain complex numbers of modulus one that 
\begin{equation}
  |h_\alpha| = \left|\mathrm{det\,}
\begin{pmatrix}
  1 & 1 &1 &  \cdots & 1\\
  1 & \mathrm{e}^{\frac{\pi i}{n}} & \mathrm{e}^{\frac{2\pi i}{n}} &  \cdots & \mathrm{e}^{\frac{(n-1)\pi i}{n}}\\
  1 & \mathrm{e}^{\frac{2\pi i}{n}} & \mathrm{e}^{\frac{4\pi i}{n}} &  \cdots & \mathrm{e}^{\frac{(n-1)2\pi i}{n}}\\
  1 & \mathrm{e}^{\frac{3\pi i}{n}} & \mathrm{e}^{\frac{6\pi i}{n}} &  \cdots & \mathrm{e}^{\frac{(n-1)3\pi i}{n}}\\
  \vdots & \vdots & \ddots & \vdots\\
  1 & \mathrm{e}^{\frac{(n-1)\pi i}{n}} & \mathrm{e}^{\frac{(n-1)2\pi i}{n}}  &\cdots & \mathrm{e}^{\frac{(n-1)^2 \pi i}{n}}
\end{pmatrix}\right|\,. \label{eq:vandermonde}
\end{equation}
The matrix under the determinant is now a Vandermonde matrix, for which there exists a simple closed form for the determinant.

\begin{lemma}\label{lem:vandermonde}
Let $V$ be a Vandermonde matrix of the form
\[
 V=
 \begin{pmatrix}1 & \omega_1 & \omega_1^2 & \cdots & \omega_1^{n-1}\eqskip
 1 & \omega_2 & \omega_2^2 & \cdots & \omega_2^{n-1} \eqskip
 1 & \omega_3 & \omega_3^2 & \cdots & \omega_3^{n-1} \eqskip
 \vdots & \vdots & \vdots & \ddots & \vdots\\1 & \omega_n & \omega_n^2 & \cdots & \omega_n^{n-1}
 \end{pmatrix}.
\]
Then its determinant is given by $\mathrm{det\,}V = \dprod_{1\leq k < j \leq n}(\omega_j -\omega_k)$.
\end{lemma}
\begin{proof}
See e.g. \cite[eq. (14.22)]{Ald} or \cite[Exercise 37]{Knu}.
\end{proof}

\begin{lemma}\label{lem:halpha}
$$
  |h_\alpha| = 2^{\frac{n(n-1)}{2}}  \prod_{j=1}^{n-1}\sin^{n-j}{\left(\frac{j\pi}{2n}\right)}\,.
$$
\end{lemma}
\begin{proof}
In our case, $\omega_j = \mathrm{e}^{\frac{i\pi (j-1)}{n}}$. Hence, applying Lemma~\ref{lem:vandermonde} we find
\[
\begin{array}{lll}
  |h_\alpha| & = & \left|\dprod_{1\leq k < j \leq n}\left(\mathrm{e}^{\frac{i\pi (j-1)}{n}}
  -\mathrm{e}^{\frac{i\pi (k-1)}{n}}\right)\right| \eqskip
  & = & \left|2^{n-1}\sin{\left(\frac{(n-1)\pi}{2n}\right)}\sin{\left(\frac{(n-2)\pi}{2n}\right)}\dots
  \sin{\left(\frac{\pi}{2n}\right)}\right. \eqskip
  & & \hspace*{1cm}\left.\times 2^{n-2}\sin{\left(\frac{(n-2)\pi}{2n}\right)}\sin{\left(\frac{(n-3)\pi}{2n}\right)}\dots
 \sin{\left(\frac{\pi}{2n}\right)}\right.\eqskip
 & & \hspace*{15mm} \dots\eqskip
 & & \hspace*{20mm} \left. \times 2^1 \sin{\left(\frac{\pi}{2n}\right)}\right|
  \eqskip
  & = & 2^{\frac{n(n-1)}{2}}\dprod_{k=1}^{n-1}\dprod_{j=1}^k\sin{\left(\frac{j\pi}{2n}\right)} \eqskip
  & = & 2^{\frac{n(n-1)}{2}}  \dprod_{j=1}^{n-1}\sin^{n-j}{\left(\frac{j\pi}{2n}\right)} \,.
\end{array}
\]
\end{proof}

\begin{remark}
The matrix under the determinant of \eqref{eq:vandermonde} resembles the discrete Fourier transform matrix (DFT-matrix).
The difference is in the exponent of exponentials, in our case the entry in the second row and the second column is
$\mathrm{e}^\frac{i\pi}{n}$, the DFT matrix has $\mathrm{e}^{-\frac{2i\pi}{n}}$. One would obtain the DFT matrix in
this step (up to normalization and complex conjugation) if one prescribed Navier boundary conditions instead. 
\end{remark}

\section{Determinant of $BY(T)-C$}
\begin{lemma}\label{lem:dett}
$$
\mathrm{det\,}
\begin{pmatrix}
\frac{1}{n!}T^n & \frac{1}{(n+1)!}T^{n+1} & \cdots & \frac{1}{(2n-1)!}T^{2n-1}\\
\frac{1}{(n-1)!}T^{n-1} & \frac{1}{n!}T^{n} & \cdots & \frac{1}{(2n-2)!}T^{2n-2}\\
\vdots & \vdots & \ddots & \vdots\\
\frac{1}{2!}T^2 & \frac{1}{3!}T^3 & \cdots & \frac{1}{(n+1)!}T^{n+1} \\
T & \frac{1}{2!}T^2 & \cdots & \frac{1}{n!}T^n 
\end{pmatrix}= T^{n^2} \prod_{j=0}^{n-1} \frac{j!}{(n+j)!}
$$
\end{lemma}
\begin{proof}
We multiply the $i$-th row of the matrix by $\frac{1}{T^{n-i+1}(i-1)!}$ and the $j$-th column by $\frac{(n+j-1)!}{T^{j-1}}$
to obtain
\begin{multline*}
  \mathrm{det\,}
\begin{pmatrix}
\frac{1}{n!}T^n & \frac{1}{(n+1)!}T^{n+1} & \cdots & \frac{1}{(2n-1)!}T^{2n-1}\\
\frac{1}{(n-1)!}T^{n-1} & \frac{1}{n!}T^{n} & \cdots & \frac{1}{(2n-2)!}T^{2n-2}\\
\vdots & \vdots & \ddots & \vdots\\
\frac{1}{2!}T^2 & \frac{1}{3!}T^3 & \cdots & \frac{1}{(n+1)!}T^{n+1} \\
T & \frac{1}{2!}T^2 & \cdots & \frac{1}{n!}T^n 
\end{pmatrix}= T^{n^2} \prod_{j=1}^{n} \fr{(j-1)!}{(n+j-1)!}
\\
\hspace*{2mm}\times
\mathrm{det\,}
\begin{pmatrix}
{n\choose 0} & {n+1 \choose 0} & \cdots &{2n-1\choose 0}\\
{n\choose 1} & {n+1 \choose 1} & \cdots &{2n-1\choose 1}\\
\vdots & \vdots & \ddots &\vdots\\
{n\choose n-2} & {n+1 \choose n-2} & \cdots &{2n-1\choose n-2}\\
{n\choose n-1} & {n+1 \choose n-1} & \cdots &{2n-1\choose n-1}
\end{pmatrix}
\end{multline*}
We will prove that the last determinant is equal to 1. Using the relation ${m+1\choose j} = {m\choose j}+{m\choose j-1}$
and the fact that the determinant does not change when a column is replaced by the diference between itself and another column,
we have
\begin{gather*}
\begin{array}{lll}
\det\begin{pmatrix}
{n\choose 0}  & \cdots &{2n-3\choose 0}&{2n-2\choose 0}&{2n-1\choose 0}\eqskip
{n\choose 1}  & \cdots &{2n-3\choose 1}&{2n-2\choose 1}&{2n-1\choose 1}\eqskip
{n\choose 2}  & \cdots &{2n-3\choose 2}&{2n-2\choose 2}&{2n-1\choose 2}\eqskip
\vdots & \ddots &\vdots &\vdots &\vdots\eqskip
{n\choose n-2} & \cdots &{2n-3\choose n-2}&{2n-2\choose n-2}&{2n-1\choose n-2}\eqskip
{n\choose n-1} & \cdots &{2n-3\choose n-1}&{2n-2\choose n-1}&{2n-1\choose n-1}
\end{pmatrix} & = & \det\begin{pmatrix}
{n\choose 0}  & \cdots &{2n-3\choose 0}&{2n-2\choose 0}&{2n-2\choose 0}\eqskip
{n\choose 1}  & \cdots &{2n-3\choose 1}&{2n-2\choose 1}&{2n-2\choose 1}+{2n-2\choose 0}\eqskip
{n\choose 2}  & \cdots &{2n-3\choose 2}&{2n-2\choose 2}&{2n-2\choose 2}+{2n-2\choose 1}\eqskip
\vdots & \ddots &\vdots &\vdots &\vdots\eqskip
{n\choose n-2} & \cdots &{2n-3\choose n-2}&{2n-2\choose n-2}&{2n-2\choose n-2}+{2n-2\choose n-3}\eqskip
{n\choose n-1} & \cdots &{2n-3\choose n-1}&{2n-2\choose n-1}&{2n-2\choose n-1}+{2n-2\choose n-2}
\end{pmatrix} \eqskip
& = & \det\begin{pmatrix}
{n\choose 0}  & \cdots &{2n-3\choose 0}&{2n-3\choose 0}&0\eqskip
{n\choose 1}  & \cdots &{2n-3\choose 1}&{2n-3\choose 1}+{2n-3\choose 0}&{2n-2\choose 0}\eqskip
{n\choose 2}  & \cdots &{2n-3\choose 2}&{2n-3\choose 2}+{2n-3\choose 1}&{2n-2\choose 1}\eqskip
\vdots & \ddots &\vdots &\vdots &\vdots\eqskip
{n\choose n-2} & \cdots &{2n-3\choose n-2}&{2n-3\choose n-2}+{2n-3\choose n-3}&{2n-2\choose n-3}\eqskip
{n\choose n-1} & \cdots &{2n-3\choose n-1}&{2n-3\choose n-1}+{2n-3\choose n-2}&{2n-2\choose n-2}
\end{pmatrix}\eqskip
& = & \det\begin{pmatrix}
{n\choose 0}  & \cdots &{2n-3\choose 0}&0&0\eqskip
{n\choose 1}  & \cdots &{2n-3\choose 1}&{2n-3\choose 0}&{2n-3\choose 0}\eqskip
{n\choose 2}  & \cdots &{2n-3\choose 2}&{2n-3\choose 1}&{2n-3\choose 1}+{2n-3\choose 0}\eqskip
\vdots & \ddots &\vdots &\vdots &\vdots\eqskip
{n\choose n-2} & \cdots &{2n-3\choose n-2}&{2n-3\choose n-3}&{2n-3\choose n-3}+{2n-3\choose n-4}\eqskip
{n\choose n-1} & \cdots &{2n-3\choose n-1}&{2n-3\choose n-2}&{2n-3\choose n-2}+{2n-3\choose n-3}
\end{pmatrix} \eqskip
& = & \mathrm{det\,}\begin{pmatrix}
{n\choose 0}  & \cdots &{2n-3\choose 0}&0&0\eqskip
{n\choose 1}  & \cdots &{2n-3\choose 1}&{2n-3\choose 0}&0\eqskip
{n\choose 2}  & \cdots &{2n-3\choose 2}&{2n-3\choose 1}&{2n-3\choose 0}\eqskip
\vdots & \ddots &\vdots &\vdots &\vdots\eqskip
{n\choose n-2} & \cdots &{2n-3\choose n-2}&{2n-3\choose n-3}&{2n-3\choose n-4}\eqskip
{n\choose n-1} & \cdots &{2n-3\choose n-1}&{2n-3\choose n-2}&{2n-3\choose n-3}
\end{pmatrix}\,.
\end{array}
\end{gather*}
We continue until we obtain the lower triangular matrix
$$
\begin{pmatrix}
{n\choose 0}  & 0 & 0 &  \cdots & 0\eqskip
{n\choose 1}  & {n\choose 0}  & 0 & \cdots & 0\eqskip
{n\choose 2}  & {n\choose 1}  & {n\choose 0} &  \cdots & 0\eqskip
\vdots & \vdots & \vdots & \ddots &\vdots\eqskip
{n\choose n-2} & {n\choose n-3} & {n\choose n-4} & \cdots & 0 \eqskip
{n\choose n-1} & {n\choose n-2} & {n\choose n-3} & \cdots & {n\choose 0} 
\end{pmatrix},
$$
whose determinant is equal to one.
\end{proof}

\section{Lower bound on $-\log{F(n)}$}
Our aim now is to compute the asymptotics of $|h_\alpha|$ for large $n$ which we will do by obtaining sufficiently sharp lower and upper bounds.
Let us denote
$$
  F(n) := 2^{\frac{n(n-1)}{2}}\prod_{j = 1}^{n-1}\sin^{n-j}{\left(\frac{j\pi}{2n}\right)}\,.
$$
Then we have the following lower bound.

\begin{theorem}\label{thm:lb}
The function $F$ satisfies
$$
  -\log{F(n)} \geq \frac{7\zeta (3)}{4\pi^2}n^2-\frac{1}{2}n\log{n}+n\left(-1+\frac{1}{2}\log{\pi}\right)+ \fr{1}{4}-\fr{\pi^2}{72n}-
  \fr{\pi^2}{144n^2}-\fr{\pi^4}{1080n^3}+\fr{\pi^4}{2160 n^4}\,.
$$
\end{theorem}
The proof of this result will be based on the following estimate.
\begin{lemma}\label{lem:lb1}
The function $F$ satisfies
$$
 -\log{F(n)} \geq  -\int_1^n (n-x)\log{\left(\sin{\left(\frac{\pi x}{2n}\right)}\right)}\,\mathrm{d}x -\frac{1}{2}(n-1)\log{\left(\sin\frac{\pi}{2n}\right)}
 - \frac{n(n-1)}{2}\log{2}\,.
$$
\end{lemma}
\begin{proof}
\begin{figure}
\centering
  \includegraphics[width=13cm]{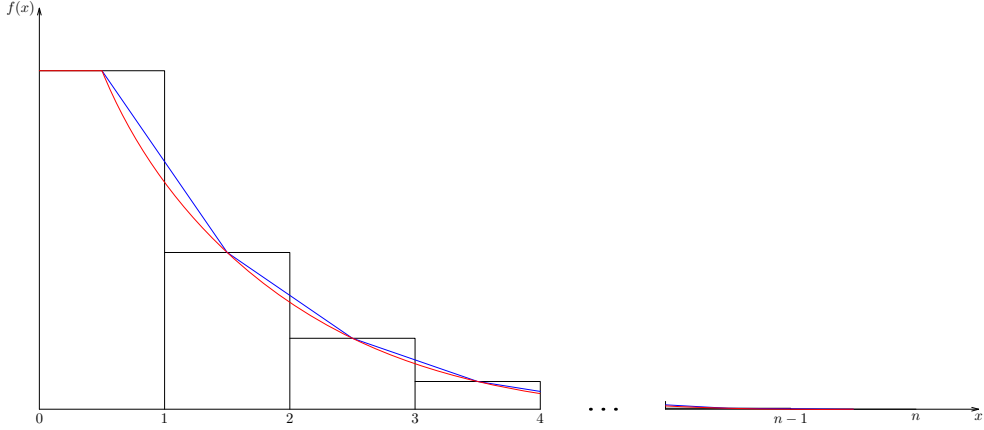}
  \caption{Illustration to the proof of Lemma~\ref{lem:lb1}. The blue line connects the point $(0,a_1)$ and the points $(j-\frac{1}{2},a_j)$, the
  red line is the function $-(n-x-\frac{1}{2})\log{\left(\sin{\left(\frac{(x+\frac{1}{2})\pi}{2n}\right)}\right)}$.}
  \label{fig1}
\end{figure}%

Applying negative logarithm to $F(n)$ we get the sum
\[
\dsum_{j=1}^n a_j = -\dsum_{j=1}^{n}(n-j)\log{\left(\sin{\left(\frac{j\pi}{2n}\right)}\right)}
\]
minus the term $\frac{n(n-1)}{2}\log{2}$. Our aim is to bound the mentioned series from below. Let us define a broken line (in blue in Figure~\ref{fig1})
connecting the point $(0,a_1)$ and the points $(j-\frac{1}{2},a_j)$. One can easily see that the area below this broken line is equal to the considered sum.
Since the function $-(n-x-\frac{1}{2})\log{\left(\sin{\left(\frac{(x+\frac{1}{2})\pi}{2n}\right)}\right)}$ (in red in Figure~\ref{fig1}) is convex, it lies
below the blue line and the corresponding area between this line and the horizontal axis is smaller than the value of the sum. Clearly, $a_n=0$. Hence the
area below the red line is given by
\begin{multline*}
  -\int_{\frac{1}{2}}^{n-\frac{1}{2}} \left(n-x-\frac{1}{2}\right)\log{\left(\sin{\left(\frac{(x+\frac{1}{2})\pi}{2n}\right)}\right)}\,\mathrm{d}x
  -\frac{1}{2}(n-1)\log{\left(\sin\frac{\pi}{2n}\right)} =
\\
 =  -\int_1^n (n-x)\log{\left(\sin{\left(\frac{\pi x}{2n}\right)}\right)}\,\mathrm{d}x -\frac{1}{2}(n-1)\log{\left(\sin\frac{\pi}{2n}\right)}\,.
\end{multline*}
\end{proof}

In order to prove Theorem~\ref{thm:lb}, we will now need to evaluate several numerical series, which we do in the following lemmata.

\begin{lemma}\label{lem:lb2}
We have 
$$
  \dsum_{k=1}^\infty \left[-\frac{1}{2}-2k^2\log{\left(1-\frac{1}{4k^2}\right)}\right] = \frac{7\zeta(3)}{4\pi^2} + \frac{1}{4}-\frac{1}{2}\log{2}\,.
$$
\end{lemma}
\begin{proof}
Using the Taylor expansion of the logarithm we have
\[
 \begin{array}{lll}  
  \dsum_{k=1}^\infty \left[-\frac{1}{2}-2k^2\log{\left(1-\frac{1}{4k^2}\right)}\right] & = & \dsum_{k=1}^\infty \left(-\frac{1}{2}+\dsum_{s=1}^\infty \frac{2k^2}{(2k)^{2s}s}\right)
  \eqskip
  & = & \fr{1}{2}\dsum_{k=1}^\infty \dsum_{s=2}^\infty \frac{1}{(2k)^{2s-2}s}\eqskip
  & = & \fr{1}{2}\dsum_{k=1}^\infty \dsum_{s=1}^\infty \frac{1}{(2k)^{2s}(s+1)}\,.
 \end{array}
\]
The term under the sum is positive, hence if the double sum converges, it converges absolutely and we can interchange the two sums.
\begin{multline}
\begin{array}{lll}
  \fr{1}{2}\dsum_{k=1}^\infty \dsum_{s=1}^\infty \frac{1}{(2k)^{2s}(s+1)} & = & \fr{1}{2}\dsum_{s=1}^\infty \dsum_{k=1}^\infty \frac{1}{(2k)^{2s}(s+1)} \eqskip
  & = & \dsum_{s=1}^\infty  \dsum_{k=1}^\infty \frac{2s+2-1}{(2k)^{2s}(2s+1)(2s+2)} \eqskip
  & = & \dsum_{s=1}^\infty \dsum_{k=1}^\infty \frac{1}{(2k)^{2s}(2s+1)} -\dsum_{s=1}^\infty \dsum_{k=1}^\infty  \frac{1}{(2k)^{2s}(2s+1)(2s+2)}\eqskip
  & = & \dsum_{s=1}^\infty \frac{\zeta(2s)}{2^{2s}(2s+1)} -  \dsum_{s=1}^\infty \frac{\zeta(2s)}{2^{2s}(2s+1)(2s+2)}\,.\label{eq:zeta3-1}
\end{array}
\end{multline}
The following relation was originally found by Euler \cite{Eul} (see also \cite[eq. (1.8)]{AS}).
$$
  \zeta(3) = \frac{\pi^2}{7}\left(1-4\dsum_{s=1}^\infty \frac{\zeta(2s)}{2^{2s}(2s+1)(2s+2)}\right)\,.
$$
From this it follows that
$$
  -\dsum_{s=1}^\infty \frac{\zeta(2s)}{2^{2s}(2s+1)(2s+2)} = \frac{7\zeta(3)}{4\pi^2}-\frac{1}{4}\,.
$$
Substituting this into \eqref{eq:zeta3-1} and using
\begin{equation}
  \dsum_{s=1}^\infty \frac{\zeta(2s)}{2^{2s}(2s+1)} = \frac{1}{2}-\frac{1}{2}\log{2}\label{eq:2s+1}
\end{equation}
(see e.g. \cite[eq. (5.4)]{Sri}) we obtain
$$
  \dsum_{k=1}^\infty \left(-\frac{1}{2}-2k^2\log{\left(1-\frac{1}{4k^2}\right)}\right) = \frac{7\zeta(3)}{4\pi^2}-\frac{1}{4}+\frac{1}{2}-\frac{1}{2}\log{2} = \frac{7\zeta(3)}{4\pi^2}+\frac{1}{4}-\frac{1}{2}\log{2}\,.
$$
The sum converges, hence interchanging of the sums is justified.
\end{proof}

\begin{lemma}\label{lem:lb3}
We have
$$
  2\dsum_{k=1}^\infty \left(1+k\log{\left(1-\frac{1}{2k}\right)}-k\log{\left(1+\frac{1}{2k}\right)}\right) = -1+\log{2}\,.
$$
\end{lemma}
\begin{proof}
Again using the Taylor expansion of the logarithm we obtain
\begin{gather*}
\begin{array}{lll}
  2\dsum_{k=1}^\infty \left[1+k\log{\left(1-\frac{1}{2k}\right)}-k\log{\left(1+\frac{1}{2k}\right)}\right] & = & 
  2\dsum_{k=1}^\infty \left(1+k\dsum_{s=1}^\infty \left(-\frac{1}{(2k)^s s}+\frac{(-1)^s}{(2k)^s s}\right)\right) \eqskip
  & = & 2 \dsum_{k=1}^\infty \left(1-k\dsum_{l =1}^\infty \frac{2}{(2k)^{2l-1}(2l-1)}\right)\eqskip
  & = & -2 \dsum_{k = 1}^\infty \dsum_{l=2}^\infty \frac{1}{(2k)^{2l-2}(2l-1)}\eqskip
  & = & -2 \dsum_{k=1}^\infty \dsum_{m=1}^\infty \frac{1}{(2k)^{2m}(2m+1)}\eqskip
  & = & -2 \dsum_{m=1}^\infty \dsum_{k=1}^\infty \frac{1}{(2k)^{2m}(2m+1)} \eqskip
  & = & -2 \dsum_{m=1}^\infty \frac{\zeta (2m)}{2^{2m}(2m+1)} \eqskip
  & = & -1+\log{2}\,.
\end{array}
\end{gather*}
We have changed the summation index to $l$ ($s=2l-1$) and then to $m = l-1$. Interchanging of the sums is justified by the fact that the term under the sums
is positive and hence since the sum converges, it converges absolutely. The last identity follows from~\eqref{eq:2s+1}.
\end{proof}

\begin{lemma}\label{lem:lb4}
We have
$$
  \dsum_{k=1}^\infty \log{\left(1-\frac{1}{4 k^2}\right)} = \log{2}-\log{\pi}\,.
$$
\end{lemma}
\begin{proof}
We have
$$
 \dsum_{k=1}^\infty \log{\left(1-\frac{1}{4k^2}\right)} = \log{\left[\prod_{k = 1}^\infty \left(1-\frac{1}{4k^2}\right)\right]}=
 \log{\left(\frac{2}{\pi}\sin{\frac{\pi}{2}}\right)} =\log{2}-\log{\pi},
$$
where we used the product expansion for the sine function
\begin{equation}
  \sin(x) = x \prod_{k=1}^\infty \left(1-\fr{x^2}{k^2 \pi^2}\right)\label{eq:sin}
\end{equation}
at $x = \pi/2$.
\end{proof}

\begin{lemma}\label{lem:boundlog}
We have the following bounds for $y\in \left[-\frac{1}{2},\frac{1}{2}\right]$
$$
  y-\frac{1}{2}y^2+\frac{1}{3}y^3-\frac{2}{3}y^4\leq \log{(1+y)}\leq y-\frac{1}{2}y^2+\frac{1}{3}y^3\,.
$$
\end{lemma}
\begin{proof}
Let us define 
$$
  g(y) := \log{(1+y)}-y+\frac{1}{2}y^2-\frac{1}{3}y^3\,,\quad h(y):= y-\frac{1}{2}y^2+\frac{1}{3}y^3-\frac{2}{3}y^4 - \log{(1+y)}\,.
$$
Then we have 
\[
 \begin{array}{ll}
  g'(y) = -\fr{y^3}{1+y}\,, & g''(y)=-\fr{3y^2+2y^3}{(1+y)^2}\eqskip
  h'(y) = -\fr{y^3(8y+5)}{3(1+y)}\,, & h''(y) = -\fr{y^2(8y^2+14y+5)}{(1+y)^2}\,.
 \end{array}
\]
Hence for $y\in \left[-\frac{1}{2},\frac{1}{2}\right]$ 
$$
  g(0) = g'(0) = 0\,,\quad h(0) = h'(0) = 0\,,\quad g''(y)\leq 0\,,\quad h''(y)\leq 0\,.
$$
Both functions are concave, hence the tangent at 0 lies above their graph and
$g(y)\leq 0$, $h(y)\leq 0$, for $y\in \left[-\frac{1}{2},\frac{1}{2}\right]$.
\end{proof}

\begin{proof}[Proof of Thm.~\ref{thm:lb}]
We begin by analysing the integral $\int_1^n (n-x)\log{\left(\sin{\left(\frac{\pi x}{2n}\right)}\right)}\,\mathrm{d}x $ which appears in Lemma~\ref{lem:lb1}.
Using the product development for the sine function given by~\eqref{eq:sin} we obtain
$$
   \int_1^n (n-x)\log\left[\sin{\left(\fr{\pi x}{2n}\right)}\right]\,\mathrm{d}x  = \int_1^n (n-x)\log\left(\fr{\pi x}{2n}\right)\,\mathrm{d}x 
   + \int_1^n(n-x)\dsum_{k=1}^\infty \log{\left(1-\frac{x^2}{4n^2k^2}\right)}\,\mathrm{d}x\,.
$$
We now evaluate the first of the two integrals in the right-hand side, for which we need the standard integrals
\begin{eqnarray*}
\begin{array}{lll}
 \dint_1^n \log{x}\,\mathrm{d}x & = & n\log{n}-n+1\,,\eqskip
 \dint_1^n x\log{x}\,\mathrm{d}x & = & \fr{n^2}{2}\log{n}-\fr{1}{4}n^2+\fr{1}{4}\,.
\end{array}
\end{eqnarray*}
Using these expressions yields
\begin{multline}
 \begin{array}{lll}
  -\dint_1^n (n-x)\log\left(\fr{\pi x}{2n}\right)\,\mathrm{d}x & = & -\dint_1^n (n-x)\left(\log{x}+\log{\left(\fr{\pi}{2n}\right)}\right)\,\mathrm{d}x\eqskip
  & = & -n^2\log{n}+n^2-n+\fr{n^2}{2}\log{n}-\fr{1}{4}n^2+\fr{1}{4}\eqskip
  & & \hspace*{1cm} -n(n-1)\log{\left(\fr{\pi}{2n}\right)}+\fr{1}{2}(n^2-1)\log{\left(\fr{\pi}{2n}\right)}\eqskip
  & = & n^2\left(\fr{3}{4}-\fr{1}{2}\log{\fr{\pi}{2}}\right) - n\log{n}+n\left(-1+\log{\fr{\pi}{2}}\right)\eqskip
  & & \hspace*{1cm} +\fr{1}{2}\log{n}+\fr{1}{4}-\fr{1}{2}\log{\fr{\pi}{2}}\,.\label{eq:lb:first}
 \end{array}
\end{multline}

For the second integral we have, after making the substitution $y = \fr{x}{2kn}$,
\begin{gather*}
\begin{array}{lll}
\dint_{\ds 1}^{\ds n} (n-x)\log{\left(1-\fr{x^2}{4k^2n^2}\right)}\,\mathrm{d}x & = & \dint_{\frac{1}{2kn}}^{\frac{1}{2k}} 2kn^2(1-2ky)\log{\left(1-y^2\right)}\,\mathrm{d}y\eqskip
 & = &  2k n^2 \Big[y\log{(1-y^2)}-2y-\log{(1-y)}+\log{(1+y)}\eqskip
 & &\hspace*{10mm} -k(y^2-1)\log{(1-y^2)}+ky^2\Big]_{\frac{1}{2kn}}^{\frac{1}{2k}} \eqskip
 & = & n^2\Big[\fr{1}{2}\left(1+4k^2\right)\log{\left(1-\fr{1}{4k^2}\right)}-\fr{3}{2}-2k\log{\left(1-\fr{1}{2k}\right)}\eqskip
 & & \hspace*{10mm}+2k\log{\left(1+\fr{1}{2k}\right)}\Big]-\Big[\Big(n-\fr{1}{2}+2k^2 n^2+2k n^2\Big)\eqskip
 & & \hspace*{20mm}\times\log{\left(1+\fr{1}{2kn}\right)}+\Big(n-\fr{1}{2}+2k^2 n^2-2k n^2\Big)\eqskip
 & & \hspace*{30mm}\times\log{\left(1-\fr{1}{2kn}\right)}+\fr{1}{2}-2n\Big]\,.
 \end{array}
\end{gather*}
Hence we have, using Lemmata~\ref{lem:lb2}, \ref{lem:lb3} and \ref{lem:lb4}, for the first part of the sum
\begin{gather*}
 -\dsum_{k=1}^\infty \Big[\fr{1}{2}\left(1+4k^2\right)\log{\left(1-\fr{1}{4k^2}\right)}-\fr{3}{2}-2k\log{\left(1-\fr{1}{2k}\right)}+2k\log{\left(1+\fr{1}{2k}\right)}\Big] = \frac{7\zeta(3)}{4\pi^2}-\frac{3}{4}+\frac{1}{2}\log{\pi}\,.
\end{gather*}
For the second part we obtain, using Lemma~\ref{lem:boundlog},
\begin{gather*}
\begin{array}{lll}
 & &  \hspace{-13mm}\dsum_{k=1}^\infty\Big[\Big(n-\fr{1}{2}+2k^2 n^2+2k n^2\Big)\log{\left(1+\fr{1}{2kn}\right)}+\Big(n-\fr{1}{2}+2k^2 n^2-2k n^2\Big)\log{\left(1-\fr{1}{2kn}\right)}+\fr{1}{2}-2n\Big]\eqskip
  & \geq & \dsum_{k=1}^\infty \Big[\left(n-\fr{1}{2}+2k^2 n^2+2k n^2\right)\left(\fr{1}{2kn}-\fr{1}{8k^2 n^2}+\fr{1}{24 k^3 n^3}-\fr{1}{24k^4 n^4}\right)\eqskip
  & & \hspace{10mm} + \left(n-\fr{1}{2}+2k^2 n^2-2k n^2\right)\left(-\fr{1}{2kn}-\fr{1}{8k^2 n^2}-\fr{1}{24 k^3 n^3}-\fr{1}{24k^4 n^4}\right)+ \fr{1}{2}-2n\Big]\eqskip
  & = & \dsum_{k=1}^\infty \Big[-\fr{1}{6k^2 n^2}+\fr{1}{6k^2 n}-\fr{1}{4k^2 n^2}\left(n-\fr{1}{2}\right)-\fr{1}{12k^4 n^4}\left(n-\fr{1}{2}\right)\Big]\eqskip
  & = & -\fr{\pi^2}{72 n}-\fr{\pi^2}{144 n^2}-\fr{\pi^4}{1080 n^3}+\fr{\pi^4}{2160 n^4}\,.
 \end{array}
\end{gather*}
We have used the facts that the parentheses in front of the logarithms are positive and $\left|\pm \fr{1}{2kn}\right|\leq \fr{1}{2}$
and the results of the sums $\dsum_{k=1}^\infty \fr{1}{k^2} = \fr{\pi^2}{6}$ and $\dsum_{k=1}^\infty \fr{1}{k^4} = \fr{\pi^4}{90}$.
Similarly, we can also obtain the upper bound
\begin{gather*}
\begin{array}{lll}
 & &  \hspace{-13mm}\dsum_{k=1}^\infty\Big[\Big(n-\fr{1}{2}+2k^2 n^2+2k n^2\Big)\log{\left(1+\fr{1}{2kn}\right)}+\Big(n-\fr{1}{2}+2k^2 n^2-2k n^2\Big)\log{\left(1-\fr{1}{2kn}\right)}+\fr{1}{2}-2n\Big]\eqskip
  & \leq & \dsum_{k=1}^\infty \Big[\left(n-\fr{1}{2}+2k^2 n^2+2k n^2\right)\left(\fr{1}{2kn}-\fr{1}{8k^2 n^2}+\fr{1}{24 k^3 n^3}\right)\eqskip
  & & \hspace{10mm} + \left(n-\fr{1}{2}+2k^2 n^2-2k n^2\right)\left(-\fr{1}{2kn}-\fr{1}{8k^2 n^2}-\fr{1}{24 k^3 n^3}\right)+ \fr{1}{2}-2n\Big]\eqskip
  & = & \dsum_{k=1}^\infty \Big[\fr{1}{6k^2 n}-\fr{1}{4k^2 n^2}\left(n-\fr{1}{2}\right)\Big]\eqskip
  & = & -\fr{\pi^2}{72 n}+\fr{\pi^2}{48 n^2}\,.
 \end{array}
\end{gather*}
Therefore, we have
\begin{equation}
-\dsum_{k=1}^\infty\int_1^n (n-x)\log{\left(1-\frac{x^2}{4k^2n^2}\right)}\,\mathrm{d}x \geq n^2\left(\frac{7\zeta(3)}{4\pi^2}-\frac{3}{4}+\frac{1}{2}\log{\pi}\right)-\fr{\pi^2}{72 n}-\fr{\pi^2}{144 n^2}-\fr{\pi^4}{1080 n^3}+\fr{\pi^4}{2160 n^4}\,.\label{eq:lb-int2}
\end{equation}
and
\begin{equation}
-\dsum_{k=1}^\infty\int_1^n (n-x)\log{\left(1-\frac{x^2}{4k^2n^2}\right)}\,\mathrm{d}x \leq n^2\left(\frac{7\zeta(3)}{4\pi^2}-\frac{3}{4}+\frac{1}{2}\log{\pi}\right)-\fr{\pi^2}{72 n}+\fr{\pi^2}{48 n^2}\,.\label{eq:lb-int3}
\end{equation}
We thus see that the sum converges, and since the function $-(n-x)\log{\left(1-\frac{x^2}{4k^2n^2}\right)}$ is nonnegative for
all $x\in[1,n]$, $k,n\in\mathbb{N}$, it also converges absolutely. This justifies exchanging the sum and the integral.
Using equations~\eqref{eq:lb:first} and~\eqref{eq:lb-int2}, Lemma~\ref{lem:lb1} and 
\begin{equation}
\begin{array}{lll}
 -\fr{1}{2}(n-1)\log\left(\sin\left(\fr{\pi}{2n}\right)\right) & \geq &  -\fr{1}{2}(n-1)\log\left(\fr{\pi}{2n}\right)\eqskip
 & = & \fr{1}{2}n\log{n}-\fr{1}{2}n\log\left(\fr{\pi}{2}\right)-
 \frac{1}{2}\log{n}+\frac{1}{2}\log{\frac{\pi}{2}}\label{eq:lb-withoutint}
\end{array}
\end{equation}
we obtain the claim of the theorem.
\end{proof}

\section{Upper bound on $-\log{F(n)}$}
\begin{theorem}\label{thm:ub}
It holds
\begin{multline*}
\begin{array}{lll}
  -\log{F(n)} &\leq &  \fr{7\zeta(3)}{4\pi^2}n^2-\fr{1}{2}n\log{n}+ \left[ \log(\pi)-\fr{1}{2}\log(2) - \fr{43}{48} \right]n +\fr{1}{12}\log{n}\eqskip
 & & \hspace{5mm}+\fr{1}{2}\left[\fr{3}{16}-\fr{\gamma}{12}-\log\left(\fr{\pi}{2}\right)\right]
 + \fr{1}{24n}\left(\fr{13}{4}-\fr{\pi^2}{3}\right)+\fr{\pi^2}{48n^2},
 \end{array}
\end{multline*}
where $\gamma$ is the Euler-Mascheroni constant.
\end{theorem}

Our aim is to bound the sum $\dsum_{j=1}^n a_j = -\dsum_{j=1}^{n}(n-j)\log{\left(\sin{\left(\frac{j\pi}{2n}\right)}\right)}$
from above. We will use the Euler-Maclaurin series
\begin{equation}
  \dsum_{j=1}^n f(j) = \int_1^n f(x)\,\mathrm{d}x + \frac{f(1)+f(n)}{2} + \int_1^n f'(x) P_1(x)\,\mathrm{d}x\,,\label{eq:ub1}
\end{equation}
where $f(x) = -(n-x)\log{\left(\sin{\left(\frac{\pi x}{2n}\right)}\right)}$ and $P_1(x) = x-\lfloor x \rfloor-\frac{1}{2}$ is
the first Bernoulli polynomial in $x-\lfloor x \rfloor$,
where $\lfloor x \rfloor$ is the largest integer smaller or equal to $x$.
From the proof of Theorem~\ref{thm:lb} we already have upper bounds for two terms on the right-hand side of~\eqref{eq:ub1},
and so it only remains to bound the third term.

\begin{lemma}\label{lem:ub1}
It holds
$$
  \dint_1^n f'(x) P_1(x)\,\mathrm{d}x \leq \fr{5}{48}n +
  \fr{1}{12}\log(n)- \fr{1}{8}\left(\fr{5}{4} +\fr{1}{3}\gamma\right)
  +\fr{13}{96n}\,.
$$
\end{lemma}
\begin{proof}
We are interested in the integral
$$
\begin{array}{lll}
   \dint_1^n f'(x) P_1(x)\,\mathrm{d}x & = & \dint_1^n \left[\log{\left(\sin{\left(\frac{\pi x}{2n}\right)}\right)}-(n-x)\frac{\pi}{2n}\cot{\left(\frac{\pi x}{2n}\right)}\right]\left(x-\lfloor x \rfloor-\frac{1}{2}\right)\,\mathrm{d}x 
\\
  & = & \dsum_{j = 1}^{n-1}\int_{j}^{j+1}  \left[\log{\left(\sin{\left(\frac{\pi x}{2n}\right)}\right)}-(n-x)\frac{\pi}{2n}\cot{\left(\frac{\pi x}{2n}\right)}\right] \left(x-j-\frac{1}{2}\right)\,\mathrm{d}x \,.
\end{array}
$$
Both the sine function on the interval $(0,\pi/2)$ and the logarithm on all of its domain are increasing functions, and we thus have 
$$
  \log{\left(\sin{\left(\fr{\pi x}{2n}\right)}\right)} \leq \log{\left(\sin{\left(\fr{(j+1)\pi}{2n}\right)}\right)}\,,\quad 
  -\log{\left(\sin{\left(\fr{\pi x}{2n}\right)}\right)} \leq -\log{\left(\sin{\left(\fr{j\pi}{2n}\right)}\right)}
$$
for $x\in [j,j+1]$. The cotangent is a decreasing function on $(0,\pi/2)$ and so
$$
  \cot{\left(\fr{\pi x}{2n}\right)} \leq \cot{\left(\fr{j\pi}{2n}\right)}\,,\quad - 
  \cot{\left(\fr{\pi x}{2n}\right)} \leq - \cot{\left(\fr{(j+1)\pi}{2n}\right)}
$$
for $x\in [j,j+1]$. The function $n-x$ is always positive for $x\in (1,n)$, while the function $x-j-\fr{1}{2}$ is non-negative for
$x\in [j+\fr{1}{2},j+1]$ and non-positive for $x\in [j,j+\fr{1}{2}]$. This allows us to write
\begin{gather}
\begin{array}{lll}
\dint_1^n f'(x) P_1(x)\,\mathrm{d}x & \leq &  \dsum_{j=1}^{n-1} \left\{\int_{j+\frac{1}{2}}^{j+1} \left[ \log{\left(\sin{\left(\fr{(j+1)\pi}{2n}\right)}\right)}
-(n-x)\fr{\pi}{2n}\cot{\left(\fr{(j+1)\pi}{2n}\right)} \right] \left(x-j-\fr{1}{2}\right)\,\mathrm{d}x \right.
\eqskip
 & & \hspace{10mm} \left.+ \dint_j^{j+\frac{1}{2}} \left[ \log{\left(\sin{\left(\fr{j\pi}{2n}\right)}\right)}-(n-x)\fr{\pi}{2n}\cot{\left(\fr{j\pi}{2n}\right)}
 \right] \left(x-j-\fr{1}{2}\right)\,\mathrm{d}x \right\} \eqskip
 & = & \fr{1}{48} \dsum_{j = 1}^{n-1} \left[6\log{\left(\sin{\left(\fr{(j+1)\pi}{2n}\right)}\right)}+\fr{\pi}{2n}\cot{\left(\fr{(j+1)\pi}{2n}\right)}
 (6j-6n+5)\right.
\eqskip
 & & \hspace{10mm}\left.- 6\log{\left(\sin{\left(\fr{j\pi}{2n}\right)}\right)}+\fr{\pi}{2n}\cot{\left(\fr{j\pi}{2n}\right)}(6n-6j-1)\right]\eqskip
 & = & \fr{1}{8} \dsum_{j = 1}^{n-1} \left[\log{\left(\sin{\left(\fr{(j+1)\pi}{2n}\right)}\right)}
 - \log{\left(\sin{\left(\fr{j\pi}{2n}\right)}\right)}\right. \eqskip
 & &\hspace{10mm} +\fr{\pi}{2n}\cot{\left(\fr{(j+1)\pi}{2n}\right)}(j-n+\fr{5}{6})
 \left.+\fr{\pi}{2n}\cot{\left(\fr{j\pi}{2n}\right)}(n-j-\fr{1}{6})\right]
 \,.\label{eq:ub2}
\end{array}
\end{gather}
Using the telescoping property for sums we have
\begin{equation}
\begin{array}{lll}
 \dsum_{j=1}^{n-1} \log\left[\sin\left(\fr{(j+1)\pi}{2n}\right)\right]-\log\left[\sin\left(\fr{j\pi}{2n}\right)\right] & = &
 \log\left[\sin\left(\fr{n\pi}{2n}\right)\right]-\log\left[\sin\left(\fr{\pi}{2n}\right)\right]  \eqskip
 & = & -\log\left[\sin\left(\fr{\pi}{2n}\right)\right]\eqskip
 & \leq & \log(n).
 \label{eq:ub3}
 \end{array}
\end{equation}
We now use two inequalities for the cotangent, valid for $x\in(0,\pi/2]$, namely,
\[
 \fr{1}{x} -\fr{4}{\pi^2}x \leq \cot(x) \leq \fr{1}{x},
\]
where the first may be found in~\cite{best} and the second follows from its expansion around zero.
Using these we obtain
\begin{eqnarray}
  \frac{\pi}{2n}\cot\left[\frac{(j+1)\pi}{2n}\right]\leq \fr{\pi}{2n}\fr{2n}{(j+1)\pi} = \fr{1}{j+1}\label{eq:ub-cot1}
\end{eqnarray}
and
\begin{eqnarray}
  -\fr{\pi}{2n}\cot\left(\frac{j\pi}{2n}\right)\leq -\fr{\pi}{2n} 
  \left(\fr{2n}{j\pi}-\fr{4}{\pi^2}\fr{j\pi}{2n}\right) =
  -\fr{1}{j}+\frac{j}{n^2}\label{eq:ub-cot2}
\end{eqnarray}
where both inequalities hold for $1\leq j\leq n-1$. We thus have
\begin{equation}
\begin{array}{lll}
  \dsum_{j = 1}^{n-1} \left[\frac{5\pi}{2n}\cot{\left(\frac{(j+1)\pi}{2n}\right)} -
  \frac{\pi}{2n}\cot{\left(\frac{j\pi}{2n}\right)}\right] &\leq & 
  \dsum_{j=1}^{n-1}\left(\frac{5}{j+1}-\frac{1}{j}+\frac{j}{n^2}\right)\eqskip
  & = & 4H_{n} -\fr{9}{2}+\fr{1}{2n}
\end{array}
\end{equation}
where $H_{n}=\dsum_{j=1}^{n} \fr{1}{j}$ denotes the harmonic number. Using the bound
\[
 H_{n} \leq \log(n) + \gamma + \fr{1}{2n},
\]
where $\gamma$ is the Euler-Mascheroni constant, we obtain
\begin{equation}\label{eq:ub5}
 \dsum_{j = 1}^{n-1} \left[\frac{5\pi}{2n}\cot{\left(\frac{(j+1)\pi}{2n}\right)} -
  \frac{\pi}{2n}\cot{\left(\frac{j\pi}{2n}\right)}\right] \leq 4\log(n)+4\gamma -\fr{9}{2} +\fr{5}{2n}.
\end{equation}

Using~\eqref{eq:ub-cot1} and~\eqref{eq:ub-cot2} with the role of $j$ and $j+1$ interchanged we obtain
\begin{equation}\label{eq:ub6}
\begin{array}{lll}
  \dsum_{j=1}^{n-1}\fr{\pi}{2n}\left[\cot{\left(\frac{(j+1)\pi}{2n}\right)}-\cot{\left(\frac{j\pi}{2n}\right)}\right](j-n) & \leq &
  \dsum_{j=1}^{n-1}\left(-\fr{1}{j+1}+\fr{1}{j}-\fr{j+1}{n^2}\right)(n-j)\eqskip
  & = & \fr{5}{6}n - H_{n} -\fr{1}{2}+\fr{2}{3n} \eqskip
  & \leq & \fr{5}{6}n - \log(n) - \gamma -\fr{1}{2}+\fr{2}{3n},
\end{array}
\end{equation}
where now we used the inequality $H_{n} \geq \log(n) + \gamma$.
Using \eqref{eq:ub2},~\eqref{eq:ub3},~\eqref{eq:ub5}, and~\eqref{eq:ub6} we obtain the sought bound.
\end{proof}

\begin{proof}[Proof of Thm.~\ref{thm:ub}]
From the proof of Thm~\ref{thm:lb} (in particular equations \eqref{eq:lb:first} and~\eqref{eq:lb-int3}) we already have
\begin{equation}
\begin{array}{lll}
  \dint_1^n f(x)\,\mathrm{d}x & \leq & n^2 \left[\fr{7\zeta(3)}{4\pi^2}+\fr{1}{2}\log{2}\right]-n\log{n}+n\left(-1+
  \log{\fr{\pi}{2}}\right)+\fr{1}{2}\log{n}\eqskip
  & & \hspace*{5mm}+\fr{1}{4}-\fr{1}{2}\log{\fr{\pi}{2}}
  -\fr{\pi^2}{72n}+\fr{\pi^2}{48n^2}\,.
\end{array}\label{eq:ub7}
\end{equation}
Furthermore, by the inequality $\sin(x)\geq \fr{2}{\pi}x$, $x\in[0,\fr{\pi}{2}]$,
\begin{equation}
\begin{array}{lll}
  \fr{f(1)+f(n)}{2} & = & \fr{f(1)}{2} \eqskip
  & = & -\fr{1}{2}(n-1)\log{\left(\sin{\left(\fr{\pi}{2n}\right)}\right)}\eqskip
  & \leq & -\fr{1}{2}(n-1)\log{\fr{1}{n}}\eqskip
  & = &  \fr{1}{2}n\log{n}-\fr{1}{2}\log{n}\label{eq:ub8}
\end{array}
\end{equation}
Using equations \eqref{eq:ub1}, \eqref{eq:ub7} and \eqref{eq:ub8} and Lemma~\ref{lem:ub1} we have
\begin{multline*}
\begin{array}{lll}
  -\log{F(n)} &=& -\fr{n(n-1)}{2}\log{2} -\dsum_{j=1}^{n}(n-j)\log{\left(\sin{\left(\fr{j\pi}{2n}\right)}\right)}\eqskip
 &\leq & \fr{7\zeta(3)}{4\pi^2}n^2-\fr{1}{2}n\log{n}+ \left[ \log(\pi)-\fr{1}{2}\log(2) - \fr{43}{48} \right]n +\fr{1}{12}\log{n}\eqskip
 & & \hspace{5mm}+\fr{1}{2}\left[\fr{3}{16}-\fr{\gamma}{12}-\log\left(\fr{\pi}{2}\right)\right]
 + \fr{1}{24n}\left(\fr{13}{4}-\fr{\pi^2}{3}\right)+\fr{\pi^2}{48n^2}\,.
 \end{array}
\end{multline*}
\end{proof}

\section{Asymptotics of the determinant of $\ph$\label{sec:ulbounds}}
First, we prove the asymptotics for the product of the factorials.
\begin{lemma}\label{lem:asymfac}
We have the following bounds
\begin{eqnarray*}
  \log\left(\fr{\prod_{j=0}^{n-1} j!}{\prod_{j=0}^{n-1} (n+j)!}\right)\!\!\! &\!\! \geq \!\! &\!\!\! -n^2 \log{n} +\left(\fr{3}{2}-2\log{2}\right)n^2 -\fr{1}{12}\log{n}\eqskip
  & & \hspace*{5mm}+
  \fr{1}{12}\log{2}-\log{A}-\fr{1}{12} -\fr{1}{320n^2}\,,\eqskip
  \log\left(\fr{\prod_{j=0}^{n-1} j!}{\prod_{j=0}^{n-1} (n+j)!}\right)\!\!\! &\!\!\leq \!\!& \!\!\! -n^2 \log{n} +\left(\fr{3}{2}-2\log{2}\right)n^2 -
  \fr{1}{12}\log{n}\eqskip
  & & \hspace*{5mm} +\fr{1}{12}\log{2}-\log{A}+\fr{1}{6}+\fr{1}{320n^2}\,,
\end{eqnarray*}
where $A\approx 1.282427$ is the Glaisher-Kinkelin constant.
\end{lemma}
\begin{proof}
We will use the definition of the Barnes G-function 
$$
  G(n+1) = \prod_{j=0}^{n-1} j! \,,\quad n\in\mathbb{N}\,.
$$
Using that we obtain
$$
  \prod_{j=0}^{n-1} (n+j)! = \fr{\prod_{j=0}^{2n-1} j!}{\prod_{j=0}^{n-1} j!} = \frac{G(2n+1)}{G(n+1)}
$$
and hence
\begin{equation}
  \log\left(\fr{\prod_{j=0}^{n-1} j!}{\prod_{j=0}^{n-1} (n+j)!}\right) = 2\log{(G(n+1))}-\log{(G(2n+1))}\,.\label{eq:bar1}
\end{equation}
Our aim is to bound the Barnes G-function. We start with bounds on the Gamma function due to \cite[Thm.~8]{alzer}
\begin{eqnarray*}
  \log{\Gamma(n)} &\geq &\left(n-\fr{1}{2}\right)\log{n} -n+\fr{1}{2}\log{(2\pi)}\,,\\
  \log{\Gamma(n)} &\leq &\left(n-\fr{1}{2}\right)\log{n} -n+\fr{1}{2}\log{(2\pi)}+\fr{1}{12n}\,.
\end{eqnarray*}
Hence we obtain, using $\Gamma(n+1) = n \Gamma(n)$,
\begin{eqnarray*}
  \log{\Gamma(n+1)} &\geq &\left(n+\fr{1}{2}\right)\log{n} -n+\fr{1}{2}\log{(2\pi)}\,,\\
  \log{\Gamma(n+1)} &\leq &\left(n+\fr{1}{2}\right)\log{n} -n+\fr{1}{2}\log{(2\pi)}+\fr{1}{12n}\,.
\end{eqnarray*}
Hence for the Barnes G-function we obtain using \cite[Thm.~1.2]{nemes}
\begin{multline*}
\begin{array}{lll}
  \log{G(n+1)} &\geq & \fr{1}{4}n^2 + n\log{\Gamma(n+1)}-\left(\fr{1}{2}n(n+1)+\fr{1}{12}\right)\log{n}-\log{A}-\fr{1}{720n^2}\eqskip
  & \geq & \fr{1}{2}n^2\log{n}-\fr{3}{4}n^2+\fr{1}{2}n\log{(2\pi)}-\fr{1}{12}\log{n}-\log{A}-\fr{1}{720n^2}\,.
 \end{array}
 \end{multline*}
\begin{multline*}
\begin{array}{lll}
  \log{G(n+1)} &\leq & \fr{1}{4}n^2 + n\log{\Gamma(n+1)}-\left(\fr{1}{2}n(n+1)+\fr{1}{12}\right)\log{n}-\log{A}+\fr{1}{720n^2}\eqskip
  & \leq & \fr{1}{2}n^2\log{n}-\fr{3}{4}n^2+\fr{1}{2}n\log{(2\pi)}-\fr{1}{12}\log{n}+\fr{1}{12}-\log{A}+\fr{1}{720n^2}\,.
 \end{array}
 \end{multline*}
Here $A$ is the Glaisher-Kinkelin constant. The found asymptotics of the Barnes G-function is in good correspondence with \cite[eq.~(A.6)]{Vor}.
Substituting these bounds into \eqref{eq:bar1} we obtain the sought bounds.
\end{proof}

We may now prove Theorem~\ref{thm:h2}, which is a consequence of Theorem~\ref{thm:h} and Lemmata~\ref{lem:halpha} and \ref{lem:dett}.

To prove the second part of the theorem, we start from the expression given in Theorem~\ref{thm:h2} to obtain
$$
  \log{(\det \ph)} = n\log{n}+n\log{2}-2\log{F(n)} + n^2 \log(T) + \log\left(\fr{\prod_{j=0}^{n-1} j!}{\prod_{j=0}^{n-1} (n+j)!}\right)\,.
$$
Using in this expression the results of Theorems~\ref{thm:lb} and~\ref{thm:ub}, and Lemma~\ref{lem:asymfac}, we obtain the bounds 
\begin{multline*}
\begin{array}{lll}
  \log{(\ds \det \ph)} &\geq &  -n^2 \log{n} + \left[\fr{7\zeta(3)}{2\pi^2}+\fr{3}{2}+\log{\fr{T}{4}}\right] n^2  +
  \left[\log{(2\pi)-2}\right]n-\fr{1}{12}\log{n}\eqskip
  & & \hspace{5mm} +\fr{1}{12}\left[\log(2)+5-12\log{A}\right]-\fr{\pi^2}{36n}-\left(\fr{1}{320}+
  \fr{\pi^2}{72}\right)\fr{1}{n^2}\eqskip
  & & \hspace{10mm}-\fr{\pi^4}{540n^3}+\fr{\pi^4}{1080n^4}\,.
 \end{array}
 \end{multline*}
and
\begin{multline*}
\begin{array}{lll}
  \log{(\ds\det \ph)} &\leq &  -n^2 \log{n} + \left[\fr{7\zeta(3)}{2\pi^2}+\fr{3}{2}+\log{\fr{T}{4}}\right] n^2 
  + \left(2\log{\pi}-\fr{43}{24}\right) n \eqskip
& & + \hspace{5mm} \fr{1}{12}\log{n}+\fr{1}{12}\left[\fr{17}{4}-\gamma+13\log(2)-12\log\left(\pi A\right)\right]\eqskip
& & \hspace{10mm} +\fr{1}{12n}\left(\fr{13}{4}-\fr{\pi^2}{3}\right)+\fr{1}{24n^2}\left(\pi^2+\fr{3}{40}\right)\,,
 \end{array}
 \end{multline*}
from which the desired result follows.

\section{Polyharmonic operators with general lower order terms up to order $n$}
We shall now allow for more general polyharmonic operators with lower order terms, and show that the asymptotic behaviour of the corresponding determinant
is not altered substantially, as stated in Theorem~\ref{thm:perturb}. In order to keep notation to a minimum, we will show how such a result can be attained for an
operator with a potential, as this is sufficient to illustrate the proof in the more general case.

To this end, let us thus consider the operator acting as $H_n = (-1)^n (\partial_x)^{2n}+q(x)$ on $(0,T)$ with Dirichlet boundary conditions and the potential
$q(x) \in C^\infty([0,T])$. We will prove that the leading term of its determinant in the limit $n\to \infty$ is the same as for the operator $P_n$. We first notice
that the construction in Theorem~\ref{thm:bfk} goes through in a similar way as for the case of $P_{n}$, the only difference being the matrix $Y(T)$ in the term
$\mathrm{det}(BY(T)-C)$. We will thus begin by estimating the solutions to the corresponding Cauchy problem.

It can be easily seen that the solution $y_\ell$ to the Cauchy problem with the initial conditions $y_{\ell}^{(k)}(0) = \delta_{k\ell}$, $0\leq k, \ell \leq 2n-1$,
satisfies the integral equation
\begin{equation*}
  y_\ell (x) = \frac{x^\ell}{\ell !}+\int_0^x \frac{(-1)^{n+1}}{(2n-1)!} (x-s)^{2n-1} q(s) y_\ell(s)\,\mathrm{d}s\,,
\end{equation*}
while its $k$-th derivative ($k\leq \ell$) is a solution of
\begin{equation}
  y_\ell^{(k)}(x) = \frac{x^{\ell-k}}{(\ell-k)!} +\int_0^x \frac{(-1)^{n+1}}{(2n-1-k)!}(x-s)^{2n-1-k} q(s) y_\ell(s)\,\mathrm{d}s\,.\label{eq:pot:ykl}
\end{equation}
{We are interested only in the index values $0\leq k \leq n-1$, $n\leq \ell \leq 2n-1$, since for Dirichlet boundary conditions $\mathrm{det}(BY(T)-C)$
is expressed by the determinant of the upper right quarter of matrix $Y(T)$, as shown in equation~\eqref{eq:byc}. We now set up a standard iterative procedure
for such equations to obtain a sequence of approximations to the solutions $y_{\ell}$. More precisely, define
\begin{eqnarray}
  y_{\ell,m+1} (x)  &:=& \frac{x^\ell}{\ell !}+\int_0^x \frac{(-1)^{n+1}}{(2n-1)!} (x-s)^{2n-1} q(s) y_{\ell,m}(s)\,\mathrm{d}s\,,\quad m\in \mathbb{N}_0\label{eq:pot:ym}\\
  y_{\ell,0}(x) &:=& \frac{x^\ell}{\ell !}\, ,\nonumber
\end{eqnarray}
for which we have the following lemma.
\begin{lemma}\label{lem:ylm}
It holds
$$
  \left|y_{\ell,m}(x)-\frac{x^\ell}{\ell!}\right|\leq \sum_{p=1}^m B_p\,,
$$
where
$$
  B_p = \frac{x^{2pn+\ell}}{(2pn+\ell)!}\left(\max_{x\in [0,T]} |q(x)|\right)^p
$$
\end{lemma}
\begin{proof}
We will proceed by induction. Clearly, the lemma is satisfied for $m=0$. Let us assume that the lemma holds for $m-1$ with $m\geq 1$. We have from \eqref{eq:pot:ym}
\begin{eqnarray*}
  \left|y_{\ell,m}(x)-\frac{x^\ell}{\ell !}\right| &\leq & \frac{\max_{x\in [0,T]} |q(x)|}{(2n-1)!} \int_0^x (x-s)^{2n-1}\left|y_{\ell,m-1}(s)-\frac{s^\ell}{\ell !}\right|\,\mathrm{d}s \\
 && \hspace{5mm} +  \frac{\max_{x\in [0,T]} |q(x)|}{(2n-1)! \ell !} \int_0^x (x-s)^{2n-1}s^\ell\,\mathrm{d}s\,, 
\end{eqnarray*}
The first expression is bounded from above by
\begin{eqnarray*}
&& \hspace{-20mm}  \frac{\max_{x\in [0,T]} |q(x)|}{(2n-1)!} \int_0^x (x-s)^{2n-1} \sum_{p=1}^{m-1} B_p\,\mathrm{d}s \\
&=&  \frac{1}{(2n-1)!}  \sum_{p=1}^{m-1} \frac{(\max_{x\in [0,T]} |q(x)|)^{p+1}}{(2 pn +\ell)!}\int_0^x (x-s)^{2n-1} s^{2pn+\ell}\,\mathrm{d}s\\
&=&  \frac{1}{(2n-1)!}  \sum_{p=1}^{m-1}  \frac{(\max_{x\in [0,T]} |q(x)|)^{p+1}}{(2 pn +\ell)!} \frac{(2n-1)! (2pn+\ell)!}{(2(p+1)n+\ell)!} x^{2(p+1)n+\ell}\\
&=& \sum_{p=2}^m B_p\,.
\end{eqnarray*}
The second expression is equal to 
$$
  \frac{\max_{x\in [0,T]} |q(x)|}{(2n-1)! \ell !} \frac{\ell ! (2n-1)!}{(2n+\ell)!}x^{2n+\ell} = B_1\,.
$$
Hence the lemma holds for all $m$.
\end{proof}

By taking limits as $m\to \infty$ to both sides of the statement of Lemma~\ref{lem:ylm}, and noting that the resulting series on the right-hand
side is convergent (by d'Alembert's ratio test) we arrive at the following corollary.
\begin{corollary}\label{cor:pot}
It holds
$$
  \left|y_{\ell}(x)-\frac{x^\ell}{\ell!}\right|\leq \sum_{p=1}^\infty \frac{x^{2pn+\ell}}{(2pn+\ell)!}\left(\max_{x\in [0,T]} |q(x)|\right)^p\,.
$$
\end{corollary}
It is now possible to obtain similar bounds for the $k$-th derivatives. 
\begin{lemma}\label{lem:pot:deriv}
For $k=0,\dots, n-1$, and $n\leq \ell\leq 2n-1$, we have 
$$
\left|y_{\ell}^{(k)}(x)-\frac{x^{\ell-k}}{(\ell-k)!}\right|\leq \sum_{p=1}^\infty \frac{x^{2pn+\ell-k}}{(2pn+\ell-k)!}
\left(\max_{x\in [0,T]}|q(x)|\right)^p \leq {\rm O}(x^{\ell-k}/n^{\ell-k})\,.
$$
\end{lemma}
\begin{proof}
We use the bound coming from equation~\eqref{eq:pot:ykl}, and Corollary~\ref{cor:pot}. Since the sum is absolutely convergent,
we can interchange the sum and the integral to obtain
\begin{eqnarray*}
 \left|y_{\ell}^{(k)}(x)-\frac{x^{\ell-k}}{(\ell-k)!}\right| & \leq & \max_{x\in [0,T]} |q(x)|\int_0^x \frac{(x-s)^{2n-1-k}}{(2n-1-k)!}  \left(\left|y_\ell(s)-\frac{s^{\ell}}{\ell!}\right|+ \frac{s^{\ell}}{\ell!}\right)\,\mathrm{d}s \\
  & \leq & \int_0^x   \frac{(x-s)^{2n-1-k}}{(2n-1-k)!}\sum_{p=1}^\infty \frac{s^{2pn+\ell}}{(2pn+\ell)!}\left(\max_{x\in [0,T]} |q(x)|\right)^{p+1} \,\mathrm{d}s \\
  & & \hspace{5mm}+\frac{\max_{x\in [0,T]} |q(x)|}{(2n-1-k)!\ell!} \int_0^x (x-s)^{2n-1-k} s^{\ell}\,\mathrm{d}s  \\
  & \leq & \sum_{p=1}^\infty \frac{(\max_{x\in [0,T]} |q(x)|)^{p+1}(2n-1-k)! (2pn+\ell)!}{(2n-1-k)! (2pn+\ell)![2n(p+1)+\ell-k]!}x^{2n(p+1)+\ell-k}\\
  & & \hspace{5mm}+\frac{\max_{x\in [0,T]} |q(x)|(2n-1-k)!\ell!}{(2n-1-k)!\ell!(2n+\ell-k)!}x^{2n+\ell-k} \\
  & = & \sum_{p=1}^\infty \frac{x^{2pn+\ell-k}}{(2pn+\ell-k)!}\left(\max_{x\in [0,T]}|q(x)|\right)^p
\end{eqnarray*}
It is not difficult to see that there exists a fixed positive constant $K$ such that for all $n,p,\ell,k\in \mathbb{N}$, $\ell> k$ it holds $\left(\max_{x\in [0,T]}|q(x)|\right)^p < K^{2pn}$. Then we have (note that due to the form of the matrix $BY(T)-C$ we are interested only in $0\leq k \leq n-1$, $n\leq \ell \leq 2n-1$)
\begin{eqnarray*}
  \sum_{p=1}^\infty \frac{x^{2pn+\ell-k}}{(2pn+\ell-k)!}\left(\max_{x\in [0,T]}|q(x)|\right)^p & \leq & \sum_{p=1}^\infty \frac{K^{2pn}x^{2pn+\ell-k}}{(2pn+\ell-k)!}\\
  & = &  \sum_{p=1}^\infty \left(\prod_{r=1}^{\ell-k}\frac{x}{2pn+r}\right) \frac{(Kx)^{2pn}}{(2pn)!} \\
  & \leq &   \left(\prod_{r=1}^{\ell-k}\frac{x}{2n+r}\right) \sum_{p=1}^\infty \frac{(Kx)^{2pn}}{(2pn)!} \\
  & \leq &   \left(\prod_{r=1}^{\ell-k}\frac{x}{2n+r}\right) \sum_{p=1}^\infty \frac{(Kx)^{p}}{p!} \\
  & = &   \left(\prod_{r=1}^{\ell-k}\frac{x}{2n+r}\right) (\mathrm{e}^{Kx}-1)\,.
\end{eqnarray*}
The first factor between brackets in this last product is of order ${\rm O}(x^{\ell-k}/n^{\ell-k})$, while the factor $(\mathrm{e}^{Kx}-1)$ is independent of $n$.
\end{proof}

From Lemma~\ref{lem:pot:deriv} and Theorem~\ref{thm:bfk} the main theorem of this section follows.
\begin{theorem}\label{thm:pot:main}
The determinant of the polyharmonic operator with potential $$ H_n =(-1)^{n} (\partial_x)^{2n}+q(x)$$ defined on the interval $(0,T)$ with Dirichlet
boundary conditions and potential $q\in C^\infty([0,T])$ satisfies
\begin{equation*}
  \log{(\det H_n)} = \log{(\det \ph)}+ {\rm O}(1/n)\,,
\end{equation*}
for large values of $n$.
\end{theorem}
\begin{proof}
The only difference from the construction used in Theorem~\ref{thm:bfk} is the determinant of the matrix $BY(T)-C$, which for the operator $\ph$
was computed in Lemma~\ref{lem:dett}. We denote the matrix appearing in the corresponding determinant for the operator $H_n$ by $M_{ij}$, $1\leq i,j \leq n$.
It follows from Lemma~\ref{lem:pot:deriv} that the entries of this matrix belong to the class $M_{ij} = \frac{T^{n+j-i}}{(n+j-i)!}+
{\rm O}((T/n)^{n+j-i})$. Furthermore, we denote by $N$ the class of matrices with entries 
$$
  N_{ij} = \frac{1}{(n+j-i)!}+ {\rm O}(1/n^{n+j-i}) =\frac{1}{(n+j-i)!} \left(1+{\rm O}\left(\frac{(n+j-i)!}{n^{n+j-i}}\right)\right) \,.
$$
By a similar approach to that at the beginning of the proof of Lemma~\ref{lem:dett} we find that it holds $\det M = T^{n^2} \det N$.
Since $\frac{(j+1)!}{n^{j+1}}\leq \frac{j!}{n^j}$, the matrix $N$ belongs to a (larger) class of matrices 
$$
  \mathrm{diag\,}\left[1+{\rm O}(n!/n^n),1+{\rm O}((n-1)!/n^{(n-1)}),\dots , 1+{\rm O}(2!/n^{2}), 1+{\rm O}(1!/n)\right]\cdot Q\,,
$$
where $Q$ is a matrix with entries $Q_{ij}=\frac{1}{(n+j-i)!}$. We will prove that the determinant of the diagonal matrix multiplying $Q$
belongs to the class $1+{\rm O}(1/n)$ and hence $\det M = T^{n^2}\det N = T^{n^2}\det Q (1+{\rm O}(1/n))$. We have
\begin{eqnarray*}
  \prod_{j=1}^n \left(1+{\rm O}\left(\frac{j!}{n^j}\right)\right) & \leq & \mathrm{exp\,}\left(C_1\sum_{j=1}^n \frac{j!}{n^j}\right)\\
  & \leq &  \mathrm{exp\,}\left(\frac{C_2}{n}\sum_{j=1}^n\frac{j}{(\sqrt{2})^j}\left(\frac{j}{\sqrt{2}n}\right)^{j-1}\right)\\
  & \leq &  \mathrm{exp\,}\left(\frac{C_3}{n}\sum_{j=1}^n \left(\frac{j}{\sqrt{2}n}\right)^{j-1}\right) \\
  & \leq &  \mathrm{exp\,}\left(\frac{C_3}{n}\sum_{j=1}^n \left(\frac{1}{\sqrt{2}}\right)^{j-1}\right) \\
  & \leq &  \mathrm{exp\,}\left(C_4/n\right)\\
  & = &  1+{\rm O}(1/n)\,,
\end{eqnarray*}
where $C_j$, $j=1,\dots,4,$ are some $n$-independent constants. In the second inequality we have used
$$
  \frac{j!}{n^j} \leq \frac{e j^j \sqrt{j}}{(\mathrm{e}n)^j}\leq \mathrm{e} \left(\frac{j}{2n}\right)^j\frac{\sqrt{j}}{(\mathrm{e}/2)^j}\leq \mathrm{e}\left(\frac{j}{2n}\right)^j
$$ 
following from $j!<\mathrm{e} j^{j+1/2} \mathrm{e}^{-j}$ and $\sqrt{j}\leq (\mathrm{e}/2)^j$.
The third inequality follows from $\frac{j}{(\sqrt{2})^j}\leq \frac{2}{\mathrm{e}\log{2}}$ and the fourth simply from $j\leq n$. The fifth inequality follows because the geometric series is summable.

Having $\det M = T^{n^2}\det Q \,(1+{\rm O}(1/n))$ we conclude that $\det H_n = \det \ph (1+{\rm O}(1/n))$ and hence we obtain the given relation for the logarithms.
\end{proof}

To extend this result to operators of the form
\[
 H_{n} = (-1)^{n} (\partial_x)^{2n}+\sum_{j=0}^m q_j(x)(\partial_x)^{j}
\]
as in Theorem~\ref{thm:perturb}, it is enough to note that the derivatives of the solutions to the Cauchy problem in this
case satisfy
\begin{eqnarray*}
  y^{(k)}_\ell(x) &=& \frac{x^{\ell-k}}{(\ell-k)!}\\
 & & \hspace*{2mm}+\sum_{j=k}^m \int_0^x \int_0^{s_{j-k}}\dots \int_0^{s_1} \frac{(-1)^{n+1}}{(2n-j-1)!}(s_1-s_0)^{2n-j-1} q_j(s_0) y_\ell^{(j)}(s_0)\,\mathrm{d}s_0\dots \mathrm{d}s_{j-k}\\
 & & \hspace*{4mm} +\sum_{j=0}^{\min{(k-1,m)}} \int_0^x \frac{(-1)^{n+1}}{(2n-k-1)!}(x-s_0)^{2n-k-1}q_j(s_0) y_\ell^{(j)}(s_0)\,\mathrm{d}s_0\,.
\end{eqnarray*}
The rest of the proof is similar to the proof for the operator with the potential $q$.

\section*{Acknowledgements}
P.F. was partially supported by the Funda{\c{c}}{\~a}o para a Ci{\^e}ncia e a Tecnologia, Portugal. 
J.L. was supported by the project ``International mobilities for research activities of the University of Hradec Kr\'alov\'e'' CZ.02.2.69/\allowbreak0.0/0.0/16{\_}027/0008487.
J.L. thanks the University of Lisbon for its hospitality during his stay in Lisbon. The authors would like to thank G. Nemes for suggesting an improvement in
the constant terms in the estimates in Lemma~\ref{lem:asymfac}.

\end{document}